\documentclass{article}

\usepackage{graphics}
\usepackage{amsmath}
\usepackage{amssymb}
\usepackage{epsfig,wrapfig}
\usepackage{graphicx}
\usepackage{amsmath}
\usepackage{epsfig,wrapfig}
\usepackage{fullpage}
\usepackage{color}

\usepackage{cite}

\newtheorem{definition}{Definition}
\newtheorem{lemma}{Lemma}
\newtheorem{theorem}{Theorem}
\newtheorem{corollary}{Corollary}

\title{\LARGE \bf  Controllability-Gramian Submatrices for a Network Consensus Model}

\author{Sandip Roy and Mengran Xue\thanks{The authors are with the School of Electrical Engineering and Computer Science at Washington State University.  Please send correspondence to
\{sandip, morashu\}@wsu.edu.  The first author acknowledges support
from United States National Science Foundation Grants 1545104 and 1635184.}}

\begin{document}
\maketitle

\begin{abstract}
Principal submatrices of the controllability Gramian and their inverses are examined, for a network-consensus model with inputs at a subset of network nodes.  Specifically, several properties of the Gramian submatrices and their inverses -- including dominant eigenvalues and eigenvectors, diagonal entries, and sign patterns -- are characterized by exploiting the special doubly-nonnegative structure of the matrices.  In addition, majorizations for these properties are obtained in terms of cutsets in the network's graph, based on the diffusive form of the model.  The asymptotic (long time horizon) structure of the controllability Gramian is also analyzed. 
The results on the Gramian are used to study metrics for target control of the network-consensus model. 
\end{abstract}

\section{Introduction}

Dynamical models for consensus or synchronization in networks have been exhaustively studied \cite{chua,xiao,ren}.  One focus of this effort has been on open-loop control of the dynamics using inputs at a subset of the network's 
nodes \cite{rahmani,control1,chen1,xue,pasqualetti,dhal}. In particular, graph-theoretic necessary or sufficient conditions for controllability have been obtained, and some characterizations of the required control energy have also been obtained using analyses of the controllability Gramian.  Recently, researchers have begun to study {\em target control} of network models, wherein inputs are designed to manipulate a group of target nodes rather than the whole network \cite{camlibel,vosughitarget,guan,physics,physics2}.  The target-control problem is of practical interest in several application domains, in which stakeholders need to use limited actuation capabilities to guide a few key nodes' states.  
In parallel with the general controllability analysis for networks, the effort on target controllability has also yielded graph-theoretic conditions and analyses of metrics. These studies demonstrate that limited-energy target control of network processes is sometimes possible even when full-state control is prohibitively costly or impossible.

Target control for network models can be analyzed in terms of principal submatrices of the controllability Gramian \cite{vosughitarget,camlibel}.  Precisely, target controllability resolves to invertibility of a principal submatrix of the Gramian, while the minimum energy required to achieve a desired target state and/or guide certain state projections can be found in terms of quadratic forms of the Gramian-submatrix inverses.  Thus, the study of target controllability motivates analysis of the Gramian matrix and its principal submatrices for network consensus/synchronization models.  

Because of the relevance of Gramian matrices to network controllability as well as dual observability/estimation problems, some structural and graph-theoretic results on the full Gramian have been developed for canonical network-consensus models, as well as for other dynamical network processes \cite{dhal,summers,vosughitarget,dinuka,pasquagramian}. 
In addition, explicit formulae for the Gramian inverse in terms of the network model's spectrum have been developed, using Cauchy matrix properties \cite{dhal}.  These explicit computations give insight into the relationship between the network model's spectrum and the required control energy.

The purpose of this study is to develop new characterizations of the Gramian and its principal submatrices in the context of a canonical discrete-time network consensus model, with the goal of assessing target control metrics.  Relative to 
the earlier studies on the Gramian of network models, the main contributions of this work are to: 1) characterize principal submatrices of the Gramian and their inverse, in addition to the full Gramian matrix; 2) give new insights into the sign patterns and eigenvalues of Gramian- and Gramian- submatrix  inverses; 3) develop graph-theoretic majorizations on the Gramian's entries; and 4) assess target control metrics and optimal inputs using the results on Gramian submatrices.  The analyses primarily draw on the diffusive structure of the network model, which imposes a spatial pattern on the input response of the network.   A main result is that Gramian submatrix inverses exhibit a special sign pattern as well as a dependence on cutsets of the network graph, which allows majorization of the control energy and analysis of minimum-energy inputs.

Although our focus here is on target control, the analyses of Gramian submatrices are germane  to the control and estimation of various network dynamical processes with a diffusive or nonnegative structure. In particular, the analyses are relevant to the controllability analysis of other discrete-time models with nonnegative state matrices
and continuous-time models with Metzler state matrices 
(e.g., models for infection spread, economic systems, etc) \cite{berman,spread}.  The results also inform other problems in network estimation and control which require consideration of Gramians and Markov parameters, including observability analysis and model reduction \cite{observability,mred1}.

The article is organized as follows.  The network consensus model is presented in Section II, and analysis of target control metrics in terms of the Gramian is reviewed in Section III.  The main results on the Gramian, and their implications on target control, are developed in Section IV. Finally, the graph-theoretic characterization of target control is illustrated in an example in Section V.

\section{Model}

A network with $n$ nodes, labeled $1,\hdots, n$, is considered.  Each node $i$ has a scalar state $x_i[k]$ which evolves in discrete time ($k \in Z^+$).  A set ${\cal S}$ containing $m$ nodes, which we call the {\em source nodes}, are amenable to external actuation.  The nodes' states evolve according to:
\begin{equation}
    {\bf x}[k+1]=A {\bf x}[k]+B{\bf u}[k],
\end{equation}
where the state vector is ${\bf x}[k]=\begin{bmatrix} x_1[k] \\ \vdots \\ x_n[k] \end{bmatrix}$, $A=[a_{ij}]$ is a row-stochastic matrix ($a_{ij} \ge 0$, $\sum_{j=1}^n a_{ij}=1$ for each $i$), ${\bf u}[k]=\begin{bmatrix} u_1[k] \\ \vdots \\ u_m[k] \end{bmatrix}$ specifies the input (actuation) signals at the $m$ source nodes, and $B$ is an $n \times m$ matrix whose columns are $0--1$ indicators of the source nodes in ${\cal S}$. 

The manipulation of the states of a set ${\cal T}$ containing $p$ {\em target nodes} is of interest. The {\em target state} vector ${\bf y}[k]$, defined as containing the states of the $p$ target nodes at time $k$, can be expressed as:
\begin{equation}
{\bf y}[k]=C {\bf x}[k],
\end{equation}
where $C$ is a $p\times m$ matrix whose rows are $0--1$ indicators of the target nodes in ${\cal T}$. We refer to
the model as a whole as the {\em input-output network consensus model} or simply the {\em network model}.

A weighted digraph $\Gamma$ is defined to represent the topology of nodal interactions in the network model.  Specifically, $\Gamma$ is defined to have $n$ nodes labeled $1,\hdots, n$, which correspond to the $n$ vertices.  A directed edge is drawn from vertex $i$ to vertex $j$ if and only if $a_{ji}>0$, and the weight of the edge is set to $a_{ji}$.
We note that the graph may include self loops (i.e., edges from vertices to themselves).  The sum of the weights of the incoming edges to each vertex is $1$.  An edge in $\Gamma$ from node $i$ to node $j$ indicates that node $j$'s state at time $k+1$ is directly influenced by the node $i$'s state at time $k$.   The vertices corresponding to the source and target nodes are referred to as source and target vertices, respectively.

Throughout the article, we assume that the matrix $A$ is irreducible and aperiodic, or equivalently that the graph $\Gamma$ is ergodic.  Under this assumption, the unactuated model reaches consensus, i.e. the manifold where all nodes' states are identical is globally asymptotically stable.

\section{Preliminaries: Target Control and the Gramian}

Target control of the network model is primarily concerned with two questions: 1) deciding whether the input ${\bf u}[k]$ can be designed to guide the target state vector ${\bf y}[k]$
to a desired goal (i.e. analyzing {\em target controllability}); and 2) determining how much actuation energy or effort is needed to do so (i.e. assessing {\em target control metrics}).  Our primary focus here is on assessing target control metrics.

Formally, target controllability
is defined as follows:
\begin{definition}
The input-output network-consensus model is said to be target controllable over $[0,k_f]$ if, for any goal  $\overline{\bf y} \in {\cal R}^p$, an input signal ${\bf u}[0],\hdots, {\bf u}[k_f-1]$ can be designed to drive the network model from a relaxed initial state ${\bf x}[0]=0$ to the target state ${\bf y}[k_f]=\overline{\bf y}$.
\end{definition}

In the case that the network model is target controllable over an interval $[0,k_f]$, the 
minimum input energy required to achieve each goal state can be assessed.  This notion is formalized in the following definition:
\begin{definition}
The {\em target-control energy} for an interval $[0,k_f]$ and goal state $\overline{\bf y}$ is defined as $E(k_f,\overline{\bf y})=\min_{{\bf u}[0],\hdots, {\bf u}[k_f-1]}\sum_{i=0}^{k_f-1} {\bf u}^T[i] {\bf u}[i]$, subject to the constraint that the input sequence drives the system from a relaxed state to ${\bf y}[k_f]=\overline{\bf y}$.  We refer to an argument (input sequence) that achieves the minimum as an {\em optimal target-control input}.  
\end{definition}

Additionally, the energy required to drive projections of the target state to a unit value may be important to characterize, as an indication of the the manipulability or security of key output statistics.  This notion is formalized as follows:
\begin{definition}
The {\em projection-manipulation  energy} for an interval $[0,k_f]$ and projection vector ${\bf \alpha}$ is defined as $F(k_f,\overline{\bf \alpha})=\min_{{\bf u}[0],\hdots, {\bf u}[k_f-1]}\sum_{i=0}^{k_f-1} {\bf u}^T[i] {\bf u}[i]$, subject to the constraint that the input sequence drives the system from a relaxed state to ${\bf \alpha}^T {\bf y}[k_f]=1$.  We refer to an input sequence that achieves the minimum as an {\em optimal projection-manipulation input}.  
\end{definition}

Often, it is of interest to characterize extremal values of the target-control energy across goal states with a particular norm, or dually of the projection-manipulation energy across projection vectors of a certain norm.  In particular, the minimum values of the two metrics across goal states and  projection vectors, respectively, are indications of the security of the network model to manipulation.  These global security notions are formalized in the following definitions:

\begin{definition}
The minimum of the target-control energy over unit-two-norm goal states, i.e. $E_{min}(k_f)=\min_{\overline{\bf y} \, s.t. \, ||{\overline{\bf y}}||_2=1} E(k_f, \overline{\bf y})$, is referred to as the {\em target security} of the network model.  A goal state $\overline{\bf y}$ that achieves the minimum is denoted as $\overline{\bf y}_{min}$, and is termed a {\em minimally-secure goal}.
\end{definition}

\begin{definition}
The minimum of the projection-manipulation energy over projection vectors with unit one-norm, i.e. $F_{min}(k_f)=\min_{\overline{\bf \alpha } \, s.t. \, ||{\bf \alpha}||_1=1} F(k_f, \overline{\bf \alpha})$, is referred to as the {\em projection security} of the network model.  A projection vector $\overline{\bf \alpha}$ that achieves the minimum is denoted as $\overline{\bf \alpha}_{min}$, and is termed a {\em minimally-secure projection}.
\end{definition}

{\em Remark:}  Other norms may be used in the security definitions. We have assumed a $1$-norm constraint on the projection vector in the projection-security definition, because the weighted sum of nodal quantities is often of interest for diffusive processes.  Meanwhile, we have assumed a two-norm constrain on the goal state in the target security definition, in keeping with standard assessments of controllability/security of linear models.

Target controllability and the target-control metrics can readily be characterized in terms of principal submatrices of the controllability Gramian of the network model.  The controllability Gramian for the network model over the interval $[0,k_f]$ is given by:
\begin{equation}
   W(k_f)= \sum_{i=0}^{k_f-1} (A^i B)(A^i B)^T .
\end{equation}
We define principal submatrices of the controllability Gramian using a set ${\cal B}$ which lists a subset of the nodes $1,\hdots, n$ in the network.  The ${\cal B}$-controllability Gramian $W({\cal B},k_f)$ is defined as the principal submatrix of $W(k_f)$ in which the rows and columns indicated in ${\cal B}$ are maintained. 

The following lemma provides characterizations of target  controllability and the target control metrics in terms of the ${\cal T}$-controllability Gramian (i.e. the principal submatrix of controllability Gramian associated with the target nodes ${\cal T}$).  These results follow directly from standard analyses of output controllability \cite{output,rugh}, hence the proof is omitted.

\begin{lemma} \label{lem:1}
The input-output network-consensus model is target controllable if and only if the ${\cal T}$-controllability Gramian 
$W({\cal T},k_f)$ is invertible.  

If the network model is target controllable, then the target-control energy is given by:
\begin{equation}
E(k_f,\overline{\bf y})= \overline{\bf y}^T W({\cal T},k_f)^{-1} \overline{\bf y},
\end{equation}
and an optimal target-control input is
\begin{equation}
    \widehat{\bf u}[i]=(C A^{k_f-i-1} B)^T W({\cal T},k_f)^{-1} \overline{\bf y}
\end{equation}
for $i=0,\hdots, k_f-1$.
Further, the target security is given by:
\begin{equation}
E_{min}(k_f)= \frac{1}{\lambda_{max} (W({\cal T},k_f))},
\end{equation}
where 
$\lambda_{max}(W({\cal T},k_f))$ refers to the largest eigenvalue of matrix
$W({\cal T},k_f)$.
The minimally secure goal is given by
\begin{equation}
\overline{\bf y}_{min}={\bf v}_{max}(W({\cal T},k_f)),
\end{equation}
where ${\bf v}_{max}(W({\cal T},k_f))$ is the right eigenvector of $W({\cal T},k_f)$ associated with the largest eigenvalue.

The projection manipulation energy is given by:
\begin{equation}
F_{min}(k_f,{\bf \alpha})= \frac{1}{{\bf \alpha }^T W({\cal T},k_f) {\bf \alpha}},
\end{equation}
The projection security is given by
\begin{equation}
F_{min}(k_f)= \frac{1}{\max_i [W({\cal T},k_f)]_{i,i}},
\end{equation}
The minimally-secure projection is given by
$\overline{\bf \alpha}_{min}={\bf e}_j$, where $j=\arg \max_i [W({\cal T},k_f)]_{i,i}$.
\end{lemma}

Lemma 1 is the starting point for the main graph-theoretic and structural analyses developed in the paper.

\section{Main Results}

Per Lemma 1, target controllability and the target control metrics are tied to properties of the controllability Gramian.  Our focus here is to develop structural and
graph-theoretic results on the Gramian of the network model, with the aim of giving insights into target control.
Because a number of graph-theoretic results have already been developed for the binary question of target controllability \cite{camlibel,vosughitarget,camlibel2}, we will primarily focus on the target-control metrics.  

First, we identify some matrix-theoretic properties of Gramian submatrices and their inverses for the network model.  These properties depend on the diffusive structure of the network-consensus model, but not on the specifics of the network's topology.  

\begin{theorem}
Consider any principal submatrix of the Gramian for the network model, say $Q=W({\cal B},k_f)$.  Also, for invertible $Q$,  consider $R=Q^{-1}$.
The matrices $Q$ and $R$ have the following properties:
\begin{enumerate}
    \item $Q$ is doubly nonnegative, i.e. it is symmetric, positive semi-definite, and entry-wise nonnegative.  For sufficiently
    large $k_f$, $Q$ is entry-wise strictly positive
    \item  The eigenvalues of $Q$ are real, nonnegative, and 
    non-defective. For sufficiently large $k_f$, the largest eigenvalue $\lambda_{max}(Q)$ has algebraic multiplicity of $1$, and its associated  eigenvector ${\bf v}_{max}(Q)$ is strictly positive (to within a scale factor).
    \item $\lambda_{max}(Q) \le \lambda_{max}(W[k_f])$.  Furthermore, for sufficiently large $k_f$, the inequality is strict.
    \item The matrix $R$ is symmetric and positive definite.  For sufficiently large $k_f$,  $R$ is irreducible.
    \item  Consider any permutation $T=PRP^{-1}$ of the matrix $R$, and consider any block-partition of $T$ as $T=\begin{bmatrix} T_{11} & T_{12} \\ T_{12}^T & T_{22} \end{bmatrix}$ where $T_{11}$ is square.  Assuming that $k_f$ is sufficiently large, the matrix 
    $T_{12}$ has at least one negative entry.   
    \item  For the special case
    that the cardinality of the set ${\cal B}$ (denoted $|{\cal B}|$) is $2$, the matrix $R$ is a nonsingular $M$ matrix.
\end{enumerate}
\end{theorem}

\begin{proof}
{\em Proof of Item 1:}  The Gramian is symmetric and positive semidefinite, hence it is immediate that the principal
submatrix $Q$ is also symmetric and positive semidefinite. We characterize the signs of the entries in $Q=[q_{ij}]$ 
as follows.  First, 
these entries are expressed in terms of the impulse responses of the network model at network nodes
due to inputs at each source node.  To simplify indexing in this analysis, we assume without loss of generality that the target nodes are the nodes $1,\hdots, p$, and hence the Gramian submatrix of interest is a leading principal submatrix.  In this case,
the entry $q_{ij}$ can be written as 
$q_{ij}=\sum_{z \in {\cal S}} h_{zi,k_f}^T h_{zj,k_f}$,
where $h_{zl,k_f}=\begin{bmatrix} {\bf e}_z^T {\bf e}_l
& {\bf e}_z^T A {\bf e}_l & \hdots & {\bf e}_z^T A^{k_f-1} {\bf e}_l \end{bmatrix}$, where the notation ${\bf e}_w$ is used for a 0--1 indicator vector with $w$th entry equal to $1$.  We notice
that $h_{zl,k_f}$ encodes the impulse response at node $l$ due to an input at node $z$.  Since the matrix $A$ is 
nonnegative, it is immediate that $h_{zl,k_f}$ is nonnegative, and hence $q_{ij} \ge 0$.  From the fact
that $A$ is irreducible and aperiodic, it follows that
the final entry in the vector $h_{zl,k_f}$ is strictly positive for all $z$ and $l$, for all sufficiently large $k_f$ (see \cite{gallagher}).  Thus, $q_{ij}$ is necessarily strictly positive for sufficiently large $k_f$. 

{\em Proof of Item 2:} Since $Q$ is symmetric and positive semidefinite, it is immediate that its eigenvalues are real, nonnegative, and nondefective (i.e. each eigenvalue's algebraic and geometric multiplicities are identical).  For sufficiently large $k_f$, we have shown above that $Q$ is strictly positive.  It thus follows from the Frobenius-Perron theory that $Q$ has a dominant eigenvalue (an eigenvalue with magnitude larger than any other eigenvalue) with algebraic multiplicity $1$, whose
eigenvector is strictly positive (upon appropriate scaling). 

{\em Proof of Item 3:}  The inequality $\lambda_{max}(Q) \le 
\lambda_{max} (W[k_f])$ is an immediate consequence of the fact that $Q$ is a principal submatrix of the positive semidefinite matrix $W[k_f]$.  The strictness of the inequality for sufficiently large $k_f$ can be proved by contradiction.  If $\lambda_{max}(Q)$ was equal to $\lambda_{max}(W[k_f])$, then from the Courant-Fisher theorem, $\max_{\bf v} \frac{{\bf v}^T W[k_f] {\bf v}}{{\bf v}^T{\bf v}}$ would equal $\lambda_{max}(Q)$, and
further the argument maximizing the quadratic form would be
the dominant eigenvector of $W[k_f]$.  However, notice
that substituting ${\bf v}=\begin{bmatrix} {\bf v}_{max}(Q) \\ {\bf 0} 
\end{bmatrix}$ into the quadratic form yields $\lambda_{max}(Q)$, but this ${\bf v}$ cannot be a dominant eigenvector since it is not strictly positive.  
Hence, a contradiction is reached.

{\em Proof of 4:}  Since $Q$ is invertible, it is in fact positive definite (in addition to being symmetric and positive semidefinite).  It is immediate that $R=Q^{-1}$ 
is symmetric and positive definite.  Since $Q$ is elementwise strictly positive, it follows that $R=Q^{-1}$ is irreducible.

{\em Proof of 5:}  This result on the sign pattern of the inverse was proved for the class of doubly-positive matrices (positive-definite matrices with strictly-positive entries) by Fiedler in \cite{fiedler1},  and generalized to the class of irreducible doubly nonnegative matrices in our recent work \cite{nonnegativenote} (see Theorem 1).

{\em Proof of 6:}  $R$ is positive definite matrix, and hence its diagonal entries are positive as is its determinant.  The off-diagonal entries are seen to be negative from the matrix inversion formula for $2 \times 2$ matrices.  Thus, the matrix is an $M$ matrix. $\blacksquare$

\end{proof}

{\em Remark 1:} The characterizations in Theorem 1 crucially depend on the doubly-nonnegative structure of the Gramian, which is a consequence of the nonnegative structure of the network-consensus model (as defined by nonnegative state, input, and output matrices).  Doubly-nonnegative matrices also arise in other contexts, such as semi-definite programming and covariance-matrix analysis \cite{doublenon1,doublenon2}.

Item 5 of Theorem 1 indicates that inverses of Gramian submatrices have a sophisticated sign pattern.  While the diagonal entries of the inverse  are positive, the off-diagonal entries may be of either sign.  Item 5 indicates, however, that some of the off-diagonal entries must be negative. The pattern of nonnegative off-diagonal entries can be given a graph-theoretic interpretation, which helps to give insight into the target-control metrics.  To formalize this interpretation, it is helpful to define a graph that represents the sign pattern of the inverse of a Gramian submatrix.  Specifically, we define the {\em negative-inverse graph} for 
the invertible Gramian submatrix $W({\cal B},k_f)$ as an (unweighted, undirected) graph on $|{\cal B}|$ vertices labeled $1,\hdots,|{\cal B}|$, where the notation $|{\cal B}|$ indicates the cardinality of the set.  An edge is drawn between vertex $i$ and $j$ if $[W^{-1}]_{ij}$ is negative. The following corollary is an immediate consequence of Item 5 of Theorem 1:

\begin{corollary}
The negative-inverse graph for any invertible Gramian submatrix
$W({\cal B},k_f)$ is connected.
\end{corollary}

The matrix-theoretic properties of the Gramian's principal submatrices developed in Theorem 1 allow characterization
of the defined target-control metrics and associated optimal input signals.  
Several results on the target-control metrics are listed in the following theorem, and then interpreted in the following discussion:

\begin{theorem}
Assume that the input-output network-consensus model is target controllable.  The target-control metrics and associated optimal input signals have the following properties, for all sufficiently large $k_f$:
\begin{enumerate}
    \item The network model has a minimally-secure goal $\overline{\bf y}_{min}$ which is strictly positive.  Additionally, the optimal input signal for the minimally-secure goal is nonnegative for $k=0,\hdots. k_f-1$.
    \item  When the network has two target nodes, the target control energy satisfies $E(k_f,|\overline{\bf y}|) \le 
E(k_f,\overline{\bf y})$ for any goal state ${\bf y}$.
    \item  The projection-manipulation energy satisfies $F_{min}(k_f,|{\bf \alpha}|) \le F_{min}(k_f, {\bf \alpha})$
    for any projection vector ${\bf \alpha}$.  Further, the 
    optimal input signal for manipulation of any projection is nonnegative.
     \item  The global security metrics satisfy the inequality: $E_{min,full}(k_f)<E_{min}(k_f)<F_{min}(k_f)$,
     where $E_{min,full}(k_f)$ refers to the target-security metric when the set of target nodes ${\cal T}$ contains all nodes in the network.
\end{enumerate}
\end{theorem}

\begin{proof}
{\em Proof of Item 1:} The minimally-secure goal $\overline{\bf y}_{min}$ is the dominant eigenvector of $W({\cal T},k_f)$.  From Item 2 of Theorem 1, this dominant eigenvector is strictly positive.  From Equation 5, the optimal target-control input for this goal is 
$\widehat{\bf u}[i]=(C A^{k_f-i-1} B)^T W({\cal T},k_f)^{-1} \overline{\bf y}_{min}$.  Since $\overline{\bf y}_{min}$ is an eigenvector of $W({\cal T},k_f)$, it is also an eigenvector of $W({\cal T},k_f)^{-1}$.  Thus, $W({\cal T},k_f)^{-1} {\bf y}_{min}$ is positive, and it follows that $\widehat{\bf u}[i]$ is nonnegative.

{\em Proof of Item 2:}  The result follows immediately from the fact that $W({\cal T},k_f)^{-1}$ is an M-matrix in this case.

{\em Proof of 3:} Since $W({\cal T},k_f)$ is a nonnegative matrix, it follows that $|{\bf \alpha}^T| W({\cal T},k_f)
|{\bf \alpha}| \ge {\bf \alpha}^T W({\cal T},k_f)
{\bf \alpha}$.  The inequality on the projection-manipulation energy follows immediately. 
The nonnegativity of the input sequence then follows from a direct computation of the optimal input sequence.

{\em Proof of 4:} The inequality relating $E_{min,full}(k_f)$ and $E_{min}(k_f)$ follows from Item 3 of Theorem 1.  The inequality relating $E_{min}(k_f)$ and
$F_{min}(k_f)$ can be derived by noting that $\lambda_{max}(W({\cal T},k_f))$ majorizes the diagonal entries of $W({\cal T},k_f)$; this follows using the same argument as used to derive Item 3 of Theorem 1.  $\blacksquare$

\end{proof}

Item 1 of Theorem 2 indicates that the network model always has a minimally-secure goal (the unit-norm goal that takes the minimum energy to reach) in the positive orthant.  Further, the lowest-energy input needed to reach this goal is nonnegative.  It is worth noting, however, that the minimum-energy input needed to reach other positively-valued goals need not be positive.  

Because of the diffusive structure of the network dynamics, one might postulate that goal states in the positive orthant (i.e., goals with identically-signed entries) require less energy to achieve.
Item 2 in the theorem demonstrates that goal states in the positive orthant indeed require less energy to reach compared to their sign-reversed versions, when there are only two target nodes.  In other words, it is easier to move any pair of nodes' states to a more synchronized goal (with both nodes' goal states having the same sign), than to a comparable differentially-signed goal. The low-energy characteristic of the positive orthant results specifically form the $M$-matrix structure of the inverse Gramian submatrix in the two-target case.  However, the result does not generalize to models with more than two target nodes: as Theorem 1 indicates, the inverse Gramian submatrix has a complicated sign pattern when the target set has
more than two nodes, which means that mixed-sign goals may sometimes require less energy to reach than their positive orthant counterparts.  However, the connectedness of the negative-sign graph (Corollary 1) does indicate that low-energy control is possible for many goal states in the nonnegative orthant.

Per Item 3, projections defined by vectors in the first quadrant are easier to manipulate than their sign-reversed counterparts, regardless of the number of target states. This characterization is a direct consequence of the nonnegative form of the Gramian submatrix.  

Finally, Item 4 provides a comparison of different global security metrics. The inequalities follow immediately from the positive definiteness of the Gramian and its submatrices, however the fact that they are strict is a consequence of the doubly-nonnegative structure of the Gramian.

{\em Remark:} All goal states in the positive orthant require less energy to reach than their sign-reversed counterparts if the corresponding inverse Gramian submatrix is an M-matrix.  Per the discussion above, this is guaranteed when the network has two target nodes.  When the network has three or more target nodes, the inverse Gramian submatrix may or may not be an M-matrix.  The class of nonnegative matrices whose inverses are M-matrices has been characterized algebraically in the linear-algebra literature (see e.g. \cite{willoughby}, and these results can be brought to bear to check whether positive-orthant goal states necessarily can be reached with low energy.

Next, we study how the graph topology of the input-output network consensus model constrains the  associated Gramian submatrix and its inverse.  The main outcome of this analysis 
is that the magnitudes of the entries in the Gramian submatrix  are small if the target nodes are far from the source.  In fact, the entries decrease monotonically as the target nodes are moved further away from the source nodes, in a certain sense (related to cutsets of the network graph). Conversely, metrics related to the inverse Gramian are large if the target nodes are far from the source nodes. To formalize these notions, we find it convenient to consider vertex-cutsets in the network graph that separate the source and target vertices.  Formally, a set of vertices ${\cal C}$ in the graph $\Gamma$ (equivalently, nodes in the network) is referred to as a {\em separating cutset}, if  all directed paths between source and target nodes in $\Gamma$ pass through a vertex in ${\cal C}$.

The following theorem characterizes the principal submatrix of the Gramian associated with the network model (i.e. the ${\cal T}$-Gramian submatrix), and its inverse, in terms of a separating cutset:
\begin{theorem}
Consider the Gramian submatrix $W({\cal T},k_f)$ for the input-output network-consensus model.  Also, let ${\cal C}$
be a separating cutset of the network graph $\Gamma$.  
Then the following inequalities hold:
\begin{enumerate}
    \item $[W({\cal T},k_f)]_{ij} \le max_l [W({\cal C},k_f)]_{ll}$.  That is, all entries in the Gramian submatrix associated with the target nodes are smaller than
    at least one of the diagonal entries of the Gramian matrix
    corresponding to the separating cutset nodes.
    \item Consider any vector ${\bf \alpha}$ such that $|{\bf \alpha}|^T {\bf 1}=1$. Then  ${\bf \alpha }^T W({\cal T},k_f) {\bf \alpha} \le \max_i [W({\cal C},k_f)]_{ii}$.
    \item  $\lambda_{max} (W({\cal T},k_f) \le 
    p \max_i [W({\cal C},k_f)]_{ii}$, where $p$ is the number of target nodes.
\end{enumerate}
\end{theorem}

\begin{proof}
From the proof of Theorem 1, the diagonal entries of the Gramian $Q=W(k_f)$ can be written in terms of the impulse responses of the network model at the corresponding nodes.
Specifically, the $l$th diagonal entry can be written as $q_{ll}=\sum_{z \in {\cal S}} h_{zl,k_f}^T h_{zl,k_f}$.

To prove the theorem, we first compare $q_{ll}$ for $l \in {\cal C}$ with $q_{ll}$ for $l \in {\cal T}$.
The crux of the proof lies in recognizing that the impulse responses $h_{zl,k_f}$ for $l \in {\cal T}$ can be expressed in terms of the impulse responses $h_{zl,k_f}$ for $l \in {\cal C}$  To formalize this, let the first define a set ${\cal V}$ which contains all nodes that are isolated from the source nodes by ${\cal C}$ (note that the set ${\cal V}$ contains ${\cal T}$ as well as all other nodes separated from the source nodes by the cutset).  The vector $\widehat{\bf x}[k]$ is defined to contain the states of the nodes in ${\cal V}$.  Then notice that the response at any node $l$ within ${\cal V}$ due to an impulse input at node $z$ can be found by solving:
\begin{eqnarray}
  & &   \widehat{\bf x}[k+1]=\widehat{A} \widehat{\bf x}[k] +\widehat{B} {\bf h}_{z,k_f}[k]  \label{eq:main}\\
  & & h_{zl,k_f}[k]={\bf e}_p^T \widehat{\bf x}[k],
\end{eqnarray}
where $\widehat{A}$ is the principal submatrix of $A$ formed by maintaining the rows/columns identified in ${\cal V}$, $\widehat{B}$ is the submatrix of $A$ with rows specified by ${\cal V}$  and columns specified by ${\cal C}$, 
${\bf h}_{z,k_f}[k]$ concatenates the impulse responses $h_{zl,k_f}$ for $l \in {\cal C}$ at time $k$, and the 0--1 indicator vector ${\bf e}_p$ is a 0--1 indicator vector which selects the response at node $l$ from the vector $\widehat{\bf x}[k]$,  We note that the expression holds for $k=0,1,\hdots,k_f$.
We stress that this expression allows computation of the impulse response at any node $l$ in ${\cal T}$  without requiring tracking of the states of nodes outside ${\cal V}$,
provided that the impulse responses at nodes in ${\cal C}$ are known.  

From Equation \ref{eq:main}, it follows that 
\begin{equation}
    h_{zl,k_f}[k]={\bf e}_p^T \begin{bmatrix} \widehat{A}^{k-1}\widehat{B} & \hdots & \widehat{A} \widehat{B} & \widehat{B} \end{bmatrix} {\bf h}_{z,k_f}, \label{eq:res}
\end{equation}
where ${\bf h}_{z,k_f}=\begin{bmatrix} {\bf h}_{z,k_f}[0] \\ \vdots \\ {\bf h}_{z.k_f}[k-1] \end{bmatrix}$.  Now consider the sum of the entries in each row of the matrix $\begin{bmatrix} \widehat{A}^{k-1}\widehat{B} & \hdots & \widehat{A} \widehat{B} & \widehat{B} \end{bmatrix}$.  
The sum of each row in this matrix is less than or equal to the sum of the corresponding row in the top-right block of $\begin{bmatrix} \widehat{A} & \widehat{B} \\ 0 & I \end{bmatrix}^{k-1}$, where the dimension of the identity matrix has been chosen so that exponentiated matrix is square.   However, as the 
matrix $\begin{bmatrix} \widehat{A} & \widehat{B} \\ 0 & I \end{bmatrix}$ has unity row sums, so does its powers.  Thus, the vector  ${\bf e}_p^T \begin{bmatrix} \widehat{A} & \widehat{B} \\ 0 & I \end{bmatrix}^{k-1}$ has entries that are nonnegative and sum to less than $1$.  Considering Equation \ref{eq:res} for $k=0,1,\hdots,k_f$, we thus find that $h_{zl,k_f}$ can be found by convolving the impulse responses $h_{zi,k_f}$ for $i \in C$ with nonnegative signals, whose total sum is less than $1$, and then summing.  Further, the same convolution can be applied to find the impulse response for each source node $z \in {\cal S}$. However, it is known that convolution by a nonnegative signal whose entries sum to less than $1$ serves to decrease the energy (two-norm) of a signal.  We thus recover that $q_{ll}=\sum_{z \in {\cal S}} h_{zl,k_f}^T h_{zl,k_f}$ must be smaller for each $l \in {\cal V}$ as compared to at least one $l \in {\cal C}$.  Since ${\cal T} \in {\cal V}$,
we have thus shown that  the diagonal entries of $W({\cal T},k_f)$ are less than or equal to  $max_l [W({\cal C},k_f)]_{ll}$.  Finally, from the fact that 
$W({\cal T},k_f)$ is doubly nonnegative, it is immediate that the off-diagonal entries are less than or equal to the largest diagonal entry.  Thus, Item 1 of the theorem statement is verified.  Items 2 and 3 then follow immediately from standard properties of positive definite matrices. $\blacksquare$

\end{proof}

The graph-theoretic analyses of Gramian submatrix  properties in Theorem 3 immediately yield graph-theoretic bounds on the defined target control metrics.  In particular, the target control metrics can be majorized in terms of the energy required to manipulate the nodes on any separating cutset of the network graph.  To present these comparisons, it is helpful to explicitly define control energy metrics for the nodes on a
separating cutset.  Specifically, first consider any vertex 
$c$ contained in a separating cutset ${\cal C}$ of the network graph.  We refer to the minimum input energy required to move the state of the corresponding network node $c$ to a unity value over the interval $[0,k_f]$ (assuming that the network is initially relaxed) as $E_c(k_f)$.   In analogy with Definition
2, $E_c(k_f)$ can be formally defined as $E_c(k_f)=\min_{{\bf u}[0],\hdots, {\bf u}[k_f-1]}\sum_{i=0}^{k_f-1} {\bf u}^T[i] {\bf u}[i]$, subject to the constraint that the input sequence drives the system from a relaxed state to $x_c[k_f]=1$.  We then
define the {\em cutset-control energy} as 
$E_{\cal C}(k_f)=\min_{c \in {cal C}} E_c(k_f)$.  The target-control metrics can be majorized in terms of the cutset-control energy, as follows:
\begin{theorem}
Consider the input-output network consensus model. The following inequalities hold for the target-control metrics, For any separating cutset ${\cal C}$ of the network graph:
\begin{enumerate}
    \item The projection-manipulation energy satisfies 
    $F(k_f, {\bf \alpha}) \ge E_{\cal C}(k_f)$ for any ${\bf \alpha}$ such that $|{\bf \alpha}|^T {\bf 1}=1$.  
    \item  The projection security satisfies 
    $F_{min}(k_f) \ge E_{\cal C}(k_f)$.  
    \item  The target security satisfies
    $E_{min(k_f)} \ge \frac{1}{p} E_{\cal C}(k_f)$.
\end{enumerate}
\end{theorem}

\begin{proof}
Item 1 in the theorem follows from Item 2 of Theorem 3, together with the expression for the projection-manipulation energy in Lemma 1.  Item 2 is verified by noticing that the inequality in Item 1 holds for all ${\bf \alpha}$ with unit one-norm, and hence holds for the vector ${\bf \alpha}$ that minimizes $F(k_f, {\bf \alpha})$.  Item 3 follows from Item 3 of Theorem 3, together with  the expression for the target security in Lemma 1. $\blacksquare$
\end{proof}

Theorem 4 demonstrates that the target control metrics follow a spatial majorization, with respect to cutsets in the network graph away from the source nodes.  Specifically, the energy required to manipulate any projection of the target state is larger than the energy required to manipulate at least one of the nodes on a separating cutset. Thus, state projections become more secure (harder to manipulate) away from the source nodes.  A similar result also holds for the target security metric, but with a scale factor related to the number of nodes being manipulated.

The values of the target-control metrics for long time horizons (i.e., in the limit of large $k_f$) are of interest, since they serve as lower bounds on energy requirements for arbitrary horizons.  Because the network-consensus process naturally asymptotes to a synchronized state, one might expect that manipulation of the dynamics to a desired synchronized state can be achieved with a vanishingly-small energy requirement, given a long time horizon.  This notion can be formalized by characterizing Gramian submatrices and their spectra for large $k_f$.  This characterization of Gramian submatrices, and consequent analyses of the target-control metrics, are formalized in the following theorem:

\begin{theorem}
Consider any principal submatrix of the Gramian for the network model, say $Q=W({\cal B},k_f)$.   The matrices $Q$ has the following properties:
\begin{enumerate}
    \item $Q$ can be written as 
    $Q=(k_f \sum_{i \in {\cal S}} w_i^2) {\bf 1}{\bf 1}^T+H(k_f)$, where the absolute values of the entries in the matrix $H(k_f)$ have an upper bound that is independent of $k_f$.  In the expression, $w_i$ is $i$th entry in the left eigenvector of $A$ associated with its unity eigenvalue, where the eigenvector has
    been normalized so that its entries sum to $1$. Also,
    ${\bf 1}$ represents a vector with all entries equal to $1$, of appropriate dimension.
    \item  The dominant eigenvalue of the matrix $Q$ is given by $\lambda_{max}(Q)=|{\cal B}| k_f \sum_{i \in {\cal S}} w_i^2+\widehat{\lambda}(k_f)$, where $|\widehat{\lambda}(k_f)|$ has an upper-bound that is independent of $k_f$.  The corresponding
    dominant eigenvector is given by ${\bf v}={\bf 1}+{\bf c}(k_f)$, where each entry in $|{\bf c}(k_f)|$ is upper bounded by an asymptotically-vanishing function of $k_f$.
    \end{enumerate}
    
     Now consider target control for the input-output network-consensus model.  Provided that the model is target controllable, the target security metric is given by: $E_{min}(k_f)=\frac{1}{p k_f \sum_{i \in {\cal S}} w_i^2}+\widehat{E}(k_f)$, where $|\widehat{E}(k_f)|$ is upper-bounded by an asymptotically-vanishing function of $k_f$.
     The minimally-secure goal is given by 
     $\overline{\bf y}_{min}={\bf 1}+{\bf c}(k_f)$,
     where each entry in $|{\bf c}(k_f)|$ is upper bounded by an asymptotically-vanishing function of $k_f$.
\end{theorem}

\begin{proof}
The state matrix $A$ for the network consensus model
has a single strictly dominant eigenvalue at $1$, with a corresponding right eigenvector of ${\bf 1}$ and a left eigenvector which is entrywise strictly
positive.  Thus, from the Jordan form of $A$, it follows immediately that powers of the matrix can be expressed as $A^n={\bf 1}{\bf w}^T+Q(n)$, where
${\bf w}^T$ is the left eigenvector of $A$ associated with the $0$ eigenvalue (whose entries have been normalized to sum to $1$), and $Q(n)$ is a matrix whose entries are each upper-bounded by a decaying exponential function of $n$.  Substituting
the expression for $A^n$ into the Gramian formula,
we find that $W(k_f)=\sum_{n=0}^{k_f-1} ({\bf 1}{\bf w}^T+R(n))B B^T  ({\bf 1}{\bf w}^T+R(n))$. With some algebra, we find that $W(k_f)=k_f \sum_{i \in {\cal S}} w_i^2 {\bf 1} {\bf 1}^T +J(k_f)$, where the entries in $J(k_f)$ each have an upper bound that is independent of $k_f$.  The form of the principal 
submatrix $Q$ of the Gramian given in the theorem statement (Item 1) follows immediately.

To characterize the dominant eigenvalue/eigenvector for $Q$, notice that the matrix can be written as 
$k_f((\sum_{i \in {\cal S}} w_i^2) {\bf 1}{\bf 1}^T
+\frac{1}{k_f}H(k_f))$.  Next, notice that the 
matrix $\sum_{i \in {\cal S}} w_i^2) {\bf 1}{\bf 1}^T$ has a non-repeated eigenvalue equal to $|{\cal B}| \sum_{i \in {\cal S}} w_i^2$, with corresponding eigenvector
equal to ${\bf 1}$; the remaining eigenalues of the matrix are equal to $0$. From standard eigenvalue and eigenvector perturbation results for non-repeated eigenvalues \cite{wilkinson}, it follows that the matrix $\sum_{i \in {\cal S}} w_i^2) {\bf 1}{\bf 1}^T
+\frac{1}{k_f}H(k_f)$ has a dominant eigenvalue equal to 
$|{\cal B}| \sum_{i \in {\cal S}} w_i^2+g(k_f)$,
where $|g(k_f)|$ is upper-bounded by a function of
the form $\frac{g}{k_f}$ for some constant $g$.  The corresponding dominant eigenvector is equal to ${\bf v}={\bf 1}+{\bf c}(k_f)$, where each entry in $|{\bf c}(k_f)|$ is upper bounded by an asymptotically-vanishing function of $k_f$.  The characterization of the dominant eigenvalue and eigenvector in Item 2 of the theorem statement follows immediately.

The characterization of the target security metric then follows from Lemma 1. $\blacksquare$
\end{proof}

{\em Remark:}  The graph-theoretic analysis of Gramian submatrices in Theorem 3 is relevant to myriad techniques which use the Gramian, beyond assessment of the open-loop control energy.  For instance, the Gramian is used in several model reduction methods such as the balanced truncation algorithm \cite{mred1,gugercin}.  The graph-theoretic analysis suggestion that, if the system being reduced has a diffusive-network structure,   
these model reduction methods can be adapted to also maintain contiguous portions of the network's graph.  We leave a careful study to future work.

\section{Example}

The spatial majorization of  target-control metrics, as developed in Theorems 3 and 4, is illustrated in an example $50$-node network.  The graph $\Gamma$ for the example network was constructed by placing vertices randomly in the unit square, and connecting vertices within a certain radius.  The edge weights of the incoming edges into each vertex were selected to be identical.  

One node in the network was subjected to actuation. 
The energy required to manipulate each individual node's state over a time horizon of 200 steps, which is the inverse of the corresponding diagonal entry of the Gramian, was determined.  These energy requirements are illustrated on the network graph in Figure \ref{fig:1}. Specifically, the actuated node (in the bottom right part of the graph) is shown in red in the plot.  Also, for all nodes, the energy required for manipulation is indicated by the size of the disk at the node.  It is seen that the nodes that are close to the source (actuated) node in the graph can be manipulated with limited energy, while distant nodes require significant energy to manipulate.  This spatial growth in the energy requirement is commensurate with Theorems 3 and 4.

\begin{figure}[!htb]
      \centering
      \includegraphics[scale=0.78]{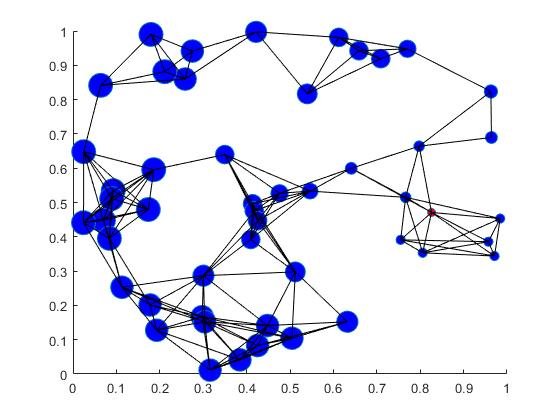} 
       \caption{The energy required to manipulate each individual node in a network-consensus process from
       a source node is shown, for a $50$-node network.  
       The source node is colored red (see bottom right part of the plot).  For each node, the size of the disk indicates the energy requirement.  Nodes near the source can be manipulated with less energy.}
      \label{fig:1}
\end{figure}

\end{document}